\documentclass{article}

\usepackage{lineno,hyperref}
\modulolinenumbers[5]

\usepackage{lipsum}
\usepackage{graphicx}
\usepackage[utf8]{inputenc}
\usepackage[T1]{fontenc}
\usepackage{soul}
\usepackage{amsthm}
\usepackage{slashbox}
\usepackage{color}
\usepackage{colortbl}
\newtheorem{teo}{Theorem}[section]

\newtheorem{lema}{Lemma}[section]

\newtheorem{cor}{Corollary}[section]
\newtheorem{cons}{Construction}[section]

\newtheorem{innercustomgeneric}{\customgenericname}
\providecommand{\customgenericname}{}

\title{The Tessellation Cover Number\\ of Good Tessellable Graphs}
\author{Alexandre Abreu\footnote{Universidade Federal do Rio de Janeiro, Brazil. Email: \{santiago, lfignacio, celina, franklin, posner\}@cos.ufrj.br}, Lu\'{\i}s Cunha$^*$, Celina de~Figueiredo$^*$, Luis Kowada\footnote{Universidade Federal Fluminense, Brazil. Email: luis@ic.uff.br}, \\Franklin Marquezino$^*$, Renato Portugal\footnote{Laborat\'orio Nacional de Computa\c{c}\~ao Cient\'ifica, Brazil. Email: portugal@lncc.br}, Daniel Posner$^*$
}

\begin{document}

\maketitle

\begin{abstract} 
A tessellation of a graph is a partition of its vertices into vertex disjoint cliques. 
  A tessellation cover of a graph is a set of tessellations that covers all of its edges, and the tessellation cover number, denoted by $T(G)$, is the size of a smallest tessellation cover. 
 The \textsc{$t$-tessellability} problem aims to decide whether a graph~$G$ has $T(G)\leq t$ and is $\mathcal{NP}$-complete for $t\geq 3$. 
  Since the number of edges of a maximum induced star of $G$, denoted by $is(G)$, is a lower bound on $T(G)$, we define good tessellable graphs as the graphs~$G$ such that $T(G)=is(G)$. 
  The \textsc{good tessellable recognition (gtr)} problem aims to decide whether $G$ is a good tessellable graph. 
  We show that \textsc{gtr} is $\mathcal{NP}$-complete not only if $T(G)$ is known or $is(G)$ is fixed, but also when the gap between $T(G)$ and $is(G)$ is large.
  As a byproduct, we obtain graph classes that obey the corresponding computational complexity behaviors.
  
\end{abstract}

\section{Introduction}\label{sec:intro}

It is known that there is a strong connection between the areas of graph theory and quantum computing. For instance, algebraic graph theory provides many tools to analyze the time-evolution of the continuous-time quantum walk, because its evolution operator is directly defined in terms of the graph's adjacency matrix. Recently, a new discrete-time quantum walk model has been defined by using the concept of graph tessellation cover~\cite{PSFG16}. Each tessellation in the cover is associated with a unitary operator and the full evolution operator is the matrix product of those operators. For practical applications, it is interesting to characterize graph classes that admit small-sized covers. Accordingly, we establish a new lower bound on tessellation cover that is described next. 

Throughout this paper we only consider undirected 
and simple graphs.
A \textit{tessellation} of a graph $G$ is a partition of its vertices into vertex disjoint cliques. 
A \textit{tessellation cover} of $G$ is a set of tessellations that covers all of its edges. 
The \textit{tessellation cover number} of $G$, denoted by $T(G)$, is the size of a smallest tessellation cover of $G$. 
If $G$ admits a tessellation cover of size $t$, then $G$ is \emph{$t$-tessellable}.
The \textsc{$t$-tessellability} problem aims to decide whether $G$ is $t$-tessellable. 
We disregard cliques of size one in a tessellation since they play no role in our proofs.  
The \textit{star number}, denoted by $is(G)$, is the number of edges of a maximum induced star of $G$. 
Notice that $T(G)\geq is(G)$, since any two edges of an induced star cannot be covered by a same tessellation. We say that $G$ is \emph{good tessellable} if $T(G) = is(G)$, and the \textsc{good tessellable recognition} (\textsc{gtr}) problem aims to decide whether a graph is good tessellable.

The known results about the tessellation cover number up to now were related to upper bounds on T(G) and the complexities of the $t$-\textsc{tessellability} problem~\cite{ArLatin,ArLAWCG,ArCNMAC}. 
Abreu et al.~\cite{ArLatin}
verified that $T(G) \leq \min \{\chi'(G),\chi(K(G))$\}, and they proved that \textsc{$t$-tessellability} is in $\mathcal{P}$ for quasi-threshold, diamond-free $K$-perfect graphs, and bipartite graphs.
On the other hand, they showed that the problem is $\mathcal{NP}$-complete for triangle-free graphs, unichord-free graphs, planar graphs with $\Delta \leq 6$, $(2,1)$-chordal graphs, $(1,2)$-graphs, and diamond-free graphs with diameter at most~five. 
Surprisingly, all the hardness results presented by Abreu et al. [2] for $t$-\textsc{tessellability} aim to decide whether $t=is(G)$, i.e., if the instance graph is good tessellable. Therefore, all their $\mathcal{NP}$-complete proofs for $t$-\textsc{tessellability} also hold for \textsc{gtr}. The only previous $\mathcal{NP}$-completeness result for \textsc{$t$-tessella\-bility} for non good tessellable graphs was presented by Posner et al.~\cite{ArCNMAC} for line graphs of triangle-free graphs (where $t = 3$ and $is(G)=2$).

We recently discovered that the concept of tessellation cover of graphs has been independently studied in the literature for a same problem, named as \textsc{equivalence covering} by Duchet~\cite{duchet1979representations} in 1979. 
Since the tessellation cover number $T(G)$ and the equivalence covering number $eq(G)$ are the same parameter, we highlight the common results, as follows: 
$\chi'(G)$ is an upper bound for $T(G)$~\cite{ArLatin} and for $eq(G)$~\cite{blokhuis1995equivalence}; 
if $G$ is triangle-free, then $T(G) = \chi'(G)$~\cite{ArLatin} and $eq(G) = \chi'(G)$~\cite{esperet2010covering}; 
if $G$ is triangle-free, then $3$-{\sc tessellability} of line graphs $L(G)$ is $\mathcal{NP}$-complete~\cite{ArCNMAC} and to decide whether $eq(G) \leq 3$ for the same class is $\mathcal{NP}$-complete as well~\cite{esperet2010covering}; 
if $G$ is $(2,1)$-chordal, then $t$-\textsc{tessellability} is $\mathcal{NP}$-complete for $t\geq 4$~\cite{ArLatin}, whereas \textsc{equivalence covering} is $\mathcal{NP}$-complete for $(1,1)$-graphs~\cite{blokhuis1995equivalence}.

\subsection*{Contributions}

We propose the \textsc{gtr} problem, which aims to decide whether a graph is good tessellable.
We analyze the combined behavior of the computational complexity of the problems \textsc{$t$-tessellability}, \textsc{gtr}, and \textsc{$k$-star size}. Clearly, these three problems belong to $\mathcal{NP}$.

\bigskip

\noindent
\begin{minipage}[c]{0.303\textwidth}
\begin{center}
        \begin{description}
            \item[] \underline{\textsc{$k$-star size}}
            \item[] \textbf{Instance:} Graph $G$ and integer $k$.
            \item[] \textbf{Question:} $is(G) \geq k$?
        \end{description}
\end{center}
\end{minipage}        
\hfill
\begin{minipage}[c]{0.303\textwidth}
\begin{center}
        \begin{description}
            \item[] \underline{\textsc{$t$-tessellability}}
            \item[] \textbf{Instance:} Graph $G$ and integer $t$.
            \item[] \textbf{Question:} $T(G) \leq t$?
        \end{description}
\end{center}
\end{minipage}        
\hfill
\begin{minipage}[c]{0.353\linewidth}      
            \begin{center}
        \begin{description}
            \item[] \underline{\textsc{gtr}}
            \item[] \textbf{Instance:} Graph $G$. \vspace{0.43cm}
            \item[] \textbf{Question:} $T(G)=is(G)$?
        \end{description}
            \end{center}
\end{minipage}

\bigskip

In order to highlight our results, we define graph classes using triples that specify the computational complexities of \textsc{$k$-star size}, \textsc{$t$-tessellability}, and \textsc{gtr}, summarized in Table~\ref{tab:classes}.

\begin{table}[!h]
\centering\footnotesize
\caption{Computational complexities of \textsc{$k$-star size}, \textsc{$t$-tessellability}, and \textsc{gtr} problems and examples of corresponding graph classes.}\label{tab:classes}
\begin{tabular}{ |c|c|c|c|p{1.1cm}| }

\hline
\backslashbox{Behavior}{Problem} & \textsc{$k$-star size} & \textsc{$t$-tessellability} & \textsc{gtr} & Examples \\ \hline
(a) & $\mathcal{P}$ & $\mathcal{NP}$-complete & $\mathcal{NP}$-complete & \cite{ArLatin, ArLAWCG}  \\ \hline
(b) & $\mathcal{P}$ & $\mathcal{NP}$-complete & $\mathcal{P}$ & \cite{ArCNMAC} Sec.~\ref{sub:23} \\ \hline
(c) & $\mathcal{NP}$-complete & $\mathcal{P}$ & $\mathcal{NP}$-complete & Sec.~\ref{sub:22} \\ \hline
(d) & $\mathcal{NP}$-complete & $\mathcal{P}$ & $\mathcal{P}$ & Sec.~\ref{sub:22} \\ \hline
(e) & $\mathcal{NP}$-complete & $\mathcal{NP}$-complete & $\mathcal{P}$ & Sec.~\ref{sub:23}\\ \hline
\end{tabular}
\end{table}

All graph classes for which Abreu et al.~\cite{ArLatin} presented hardness proofs for \textsc{$t$-tessellability} obey behavior~(a), since for those classes $is(G)$ is fixed and equal to $t$. 
The graphs studied by Posner et al.~\cite{ArCNMAC} obey behavior~(b), since for those graphs $is(G)=2$ and \textsc{$3$-tessellability} is $\mathcal{NP}$-complete. In Section~\ref{sub:23}, we present additional examples that obey behavior~(b) with $T(G)$ arbitrarily larger than a non fixed $is(G)$. Graphs of Construction~\ref{cons:2}~(I) in Section~\ref{sub:22} are examples that obey behavior~(c), since $T(G)$ is known but $k$-\textsc{star size} is $\mathcal{NP}$-complete for $k=T(G)$, which implies that \textsc{gtr} is $\mathcal{NP}$-complete. Graphs of Construction~\ref{cons:2}~(II) in Section~\ref{sub:22} are examples that obey behavior~(d), because \textsc{$k$-star size} is $\mathcal{NP}$-complete for $k = T(G)-1$, $T(G)$ is known, and $T(G) > is(G)$, which implies \textsc{gtr} is in $\mathcal{P}$. Graphs of Construction~\ref{cons:classe} in Section~\ref{sub:23} are examples that obey behavior~(e),  since it is known that $T(G) > is(G)$, which implies \textsc{gtr} is in $\mathcal{P}$, and we construct graphs so that \textsc{$k$-star size} and \textsc{$t$-tessellability} are $\mathcal{NP}$-complete.

Notice that there are omitted triples in Table~\ref{tab:classes}. Threshold graphs and bipartite graphs are examples of graph classes that obey the behavior $(\mathcal{P}, \mathcal{P}, \mathcal{P})$~\cite{ArLatin}. 
There are no graphs that obey the behavior $(\mathcal{P}, \mathcal{P}, \mathcal{NP}$-complete$)$, since if both \textsc{$k$-star size} and \textsc{$t$-tessellability} are in $\mathcal{P}$, so is \textsc{gtr}. 
Graph classes obtained by the union of graphs $G_1$ and $G_2$ so that $G_1$ is in a graph class that obey behavior (a) and $G_2$ is in a graph class that obey behavior (c) are examples satisfying the behavior $(\mathcal{NP}\textrm{-complete}, \mathcal{NP}\textrm{-complete}, \mathcal{NP}\textrm{-complete})$.

\subsection*{Notation and graph theory terminologies}

Given a graph $G=(V, E)$, 
the neighborhood $N(v)$ (or $N_G(v)$) of a vertex $v \in V$ of $G$ is given by $N(v)=\{u\ |\ uv \in E(G)\}$.
$\Delta(G)$ is the size of a maximum neighborhood of a vertex of $G$. 
We say that a vertex $u$ of $G$ is universal if $|N(u)| = |V(G)|-1$.
A graph is \textit{universal} if it has a universal vertex.
A \textit{clique} of $G$ is a subset of $V$ with all possible edges between its vertices.
An \textit{independent set} of $G$ is a subset of $V$ with no edge between any of its vertices.
A \textit{matching} of $G$ is a subset of edges of $E$  without a common endpoint.
A \textit{$k$-coloring} of $G$ is a partition of $V$ into $k$ independent sets.
A \textit{$k$-clique cover} of $G$ is a partition of $V$ into $k$ cliques.
A \textit{$k$-edge coloring} of $G$ is a partition of $E$ into $k$ matchings.

The parameters $\alpha(G)$, $\omega(G)$, and $\mu(G)$ are the size of a maximum independent set, the size of the maximum clique, and the size of the maximum matching of a graph $G$, respectively.
The \textit{chromatic number} $\chi(G)$ (\textit{chromatic index $\chi'(G)$}) is the minimum $k$ for which $G$ admits a $k$-coloring ($k$-edge coloring), and the \textit{clique cover number} $\theta(G)$ is the minimum $k$ for which $G$ admits a $k$-clique cover.
Note that $\theta(G) = \chi(G^c)$ and $\alpha(G) = \omega(G^c)$, where $G^c$ denotes the \textit{complement} of $G$ for which $V(G^c) = V(G)$ and $E(G^c)= \{xy | x \in V(G), y \in V(G), x \neq y\} \setminus E(G)$. 
The \textsc{$k$-colorability} (\textsc{$k$-edge colorability}) aims to decide whether a graph $G$ has $\chi(G)\leq k$ ($\chi'(G)\leq k$).
The \textsc{$k$-independent set} problem aims to decide whether a graph $G$ has $\alpha(G) \geq k$.

The \emph{line graph} $L(G)$ of a graph $G$ is the graph such that each edge of $E(G)$ is a vertex of $V(L(G))$, and two vertices of $V(L(G))$ are adjacent if their corresponding edges in $G$ have a common endpoint. 
The \emph{clique graph} $K(G)$ of a graph $G$ is the graph such that each maximal clique of $G$ is a vertex of $V(K(G))$, and two vertices of $V(K(G))$ are adjacent if their corresponding maximal cliques in $G$ have a common vertex. 
$S_k(G)$ is the graph obtained from $G$ by subdividing $k$ times each edge  $e= xy \in E(G)$, i.e., each edge $e=xy$ is replaced by a path $(x, v_1, v_2, \ldots, v_k, y)$.

The \emph{union} $G \cup H$ of two graphs $G$ and $H$ has $V(G \cup H) = V(G) \cup V(H)$ and $E(G \cup H) = E(G) \cup E(H)$.
The \emph{join} $G \vee H$ of two graphs $G$ and $H$ has $V(G \vee H) = V(G) \cup V(H)$ and $E(G \vee H) = E(G) \cup E(H) \cup \{vw\ |\ v \in V(G)$ and $w \in V(H)\}$.
An \textit{induced subgraph} $H=(V_H, E_H)$ of a graph $G=(V_G, E_G)$ has $V_H \subseteq V_G$ and $E_H=\{vw\ |\ v \in V(H), w \in V(H),$ and $vw \in E(G)\}$.
$G[S]$ is the induced subgraph of $G$ by the set of vertices $S \subseteq V(G)$.

\section{Graphs with known $T(G)$}
\label{sub:22}

We prove in this section that \textsc{gtr} is $\mathcal{NP}$-complete for graphs of Construction~\ref{cons:2}~(I), which have a known tessellation cover number. 
Using this result, we provide a graph class that obeys the behavior~(c) and another graph class that obeys behavior~(d). 
Note that if the tessellation cover number of $G$ is upper bounded by a constant, then we obtain $is(G)$ in polynomial time using a brute force algorithm.

The Mycielski graph $M_j$ for $j \geq 2$ has chromatic number $j$, maximum clique size $2$, and 
is defined as follows. $M_2 = K_2$ and for $j>2$, $M_j$ is obtained
from 
$M_{j-1}$
with vertices $v_1, \ldots, v_{|V(M_{j-1})|}$ by adding vertices $u_1, \ldots, u_{|V(M_{j-1})|}$ and one more vertex $w$. Each vertex $u_i$ is adjacent to all vertices of $N_{M_{j-1}}(v_i)~\cup~\{ w \}$.

\begin{cons}
\label{cons:1}
\emph{
Let $i$ be a non-negative integer and $G$ a graph.
The $(i, G)$-graph is obtained as follows.
Add $i$ vertices to graph $G$, and then add a universal vertex.
}
\end{cons}

\begin{cons}
\label{cons:2}
\emph{
Let $i$ be a non-negative integer and $G$ a graph with $V(G)=\{v_1, \ldots, v_n\}$.
We construct a graph $H=H_1 \cup H_2$ as follows.
Add $i$ disjoint copies $G_1, \ldots, G_i$ of $G$ to $H_1$, such that $V(G_j)=\{v^j_1, \ldots, v^j_n\}$ for $1 \leq j \leq i$, where $v^j_k$ represents the same vertex $v_k$ of $G$ for $1\leq k \leq n$.
Add to $H_1$ all possible edges between pairs of vertices that represent the same vertex of $G$.
Add a vertex $u$ to $H_1$ adjacent to all $v^j_k$ for $1\leq j \leq i$ and $1\leq k \leq n$.
Now, we consider two possibilities: either 
(I) $H_2$ is ($|V(G)|-3, M_3^c)$-graph of Construction~\ref{cons:1} or (II) $H_2$ is ($|V(G)|-3, M_4^c)$-graph of Construction~\ref{cons:1}.
Denote the universal vertex of $H_2$ by $u'$.
}
\end{cons}

Figure~\ref{fig:fixedtg} provides an example of a graph of Construction~\ref{cons:2} (I). 
In (a) we have an edge coloring of the graph $G \vee \{x\}$ with $|V(G)|$ colors.
In (b) we have the graph $H=H_1 \cup H_2$ and a tessellation cover of $H$ with $|V(G)|$ tessellations. 

\begin{figure}[!ht]
\centering
     \includegraphics[scale=0.23]{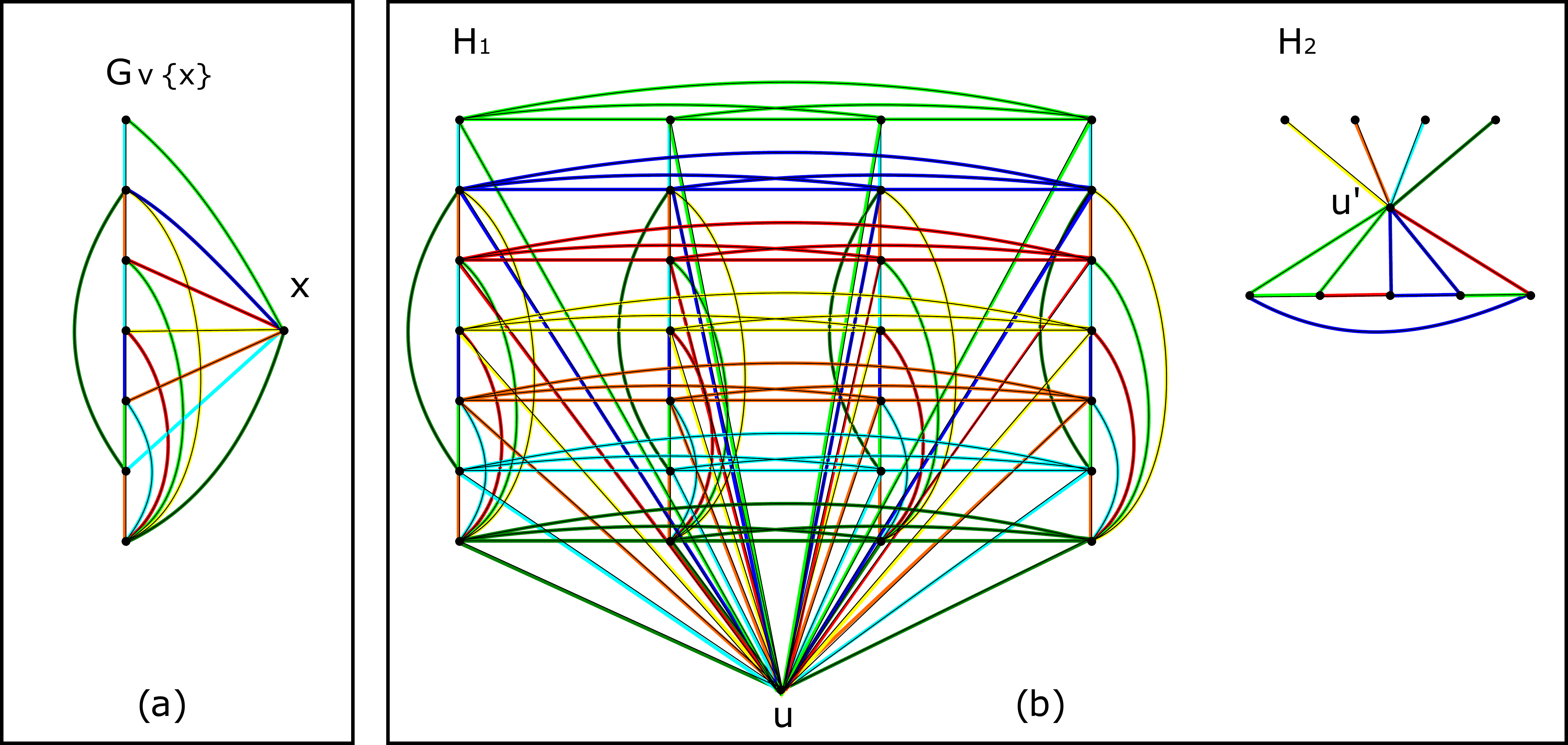}
     \caption{(a) An edge-coloring of $G\vee \{x\}$. (b) Example of a graph $H_1\cup H_2$ of Construction~\ref{cons:2} (I) obtained from graph $G$. \label{fig:fixedtg}}
\end{figure}

We now verify that the graphs of Construction~\ref{cons:2} (I) obey the behavior~(c) by 
showing that $T(H)$ is a known fraction of the number of vertices of $H$ and by deciding whether $is(H) \geq k$ is $\mathcal{NP}$-complete for $k = T(H)$.
This also implies that the graphs of Construction~\ref{cons:2} (II) obey the behavior~(d), since we have increased $T(H)$ by one unit when we have replaced $M_3^c$ by $M_4^c$ in $H_2$. In this case $T(H) > is(H)$ and \textsc{gtr} is in $\mathcal{P}$ with answer always no, whereas to decide whether $is(H) \geq k$ remains $\mathcal{NP}$-complete for $k = T(H)-1$.

\begin{teo}\label{thm:fraction}
\textsc{$k$-star size} and \textsc{gtr} are $\mathcal{NP}$-complete for graphs of Construction~\ref{cons:2}~(I).
\label{teo:npcTGfixed}
\end{teo}
\begin{proof}
Let $G$ be a graph without a universal vertex and an instance of the \textsc{$q$-colora\-bility}  problem, a well-known $\mathcal{NP}$-complete problem~\cite{ArGarey}. 
Consider the graph $H=H_1 \cup H_2$ of Construction~\ref{cons:2} (I) on $G$ with $i=q$.

We need $3$ tessellations to cover the edges of $M_3^c\vee \{u'\}$, and another $|V(G)|-3$ tessellations to cover the remaining edges of the pendant vertices,
thus, by construction, $T(H_2) = |V(G)|$.
Moreover, since $\alpha(M_3^c)=2$, then $is(H_2) = |V(G)|-1$.

We define  a tessellation cover of $H_1$ with $|V(G)|$ tessellations as follows.
Consider an optimum edge-coloring of the graph $G \vee \{x\}$.
Since $G$ has no universal vertex, $x$ is the unique universal vertex and we know that $\chi'(G \vee \{x\}) = \Delta(G \vee \{x\}) = |V(G)|$~\cite{ArClasse1}. 
Now, when we remove $x$ and the edges incident to it, this edge-coloring is a tessellation cover of $G$ with $|V(G)|$ tessellations, where for each vertex there is a distinct unused tessellation. 
We now use this tessellation cover to each copy of $G$ in $H_1$.
Next, we entirely cover each clique between vertices that represent the same vertex of $G$ and the edges incident to $u$ with the available tessellation for this clique.
Therefore, $T(H_1) \leq |V(G)|$.

We have $T(H) = \max\{T(H_1), T(H_2)\} = |V(G)| = \frac{|V(H)|-1}{q}$ and $is(H) = \max\{is(H_1), is(H_2)\}$.
Since $is(H_2) = |V(G)|-1$, $H$ is good tessellable if and only if $is(H_1) = |V(G)|$. 
Poljak~\cite{ArPoljak} proved that a graph $G$ admits a $q$-coloring if and only if $\alpha(H_1 \setminus \{u\}) = |V(G)|$.
Since $is(H) = \alpha(H_1 \setminus \{u\})$, deciding whether $H$ is good tessellable is equivalent to deciding whether $G$ is $q$-colorable.
\end{proof}

\section{Universal graphs}
\label{sub:23}

The local behavior of tessellation covers given by Lemma~\ref{lemachilocal} below motivates us to study universal graphs in this section, since the induced subgraph $G[\{v\} \cup N(v)]$ is a universal graph.
We prove that \textsc{$t$-tessellability} remains $\mathcal{NP}$-complete even if the gap between $T(G)$ and $is(G)$ is large.
Using this proof, we provide a graph class that obeys behavior~(e).

Given a $t$-tessellable graph $G$ and a vertex $v\in V(G)$, 
we consider the 
relation between $\chi(G^c[N_G(v)])$ and the cliques of those $t$ tessellations that share a vertex $v$.
Note that these cliques cover all edges incident to $v$ in any tessellation cover of $G$.
Moreover, the vertices of the neighborhood of $v$ in a same tessellation are a clique in $G$ and, therefore, they are an independent set in $G^c$. 
The independent sets in $G^c$ given by these cliques of $N_G(v)$ may share some vertices, and we can choose whichever color class they belong in such coloring of $G^c[N_G(v)]$. 
Therefore, for any vertex $v$ of~$G$, $\chi(G^c[N_G(v)]) \leq t$.
Since $is(G[v \cup N_G(v)]) = \omega(G^c[N_G(v)])$, $is(G[v \cup N_G(v)]) = \omega(G^c[N_G(v)]) \leq \chi(G^c[N_G(v)]) \leq t$, and we have the following result.

\begin{lema}
\label{lemachilocal}
If $G$ is a $t$-tessellable graph, then 
$$\max_{v \in V(G)}\{is(G[v \cup N_G(v)])\} 
\leq 
\max_{v \in V(G)}\{\chi(G^c[N_G(v)])\} \leq t.$$ 
Let $u \notin V(G)$ be a vertex. If $G \vee \{u\}$ is a $t$-tessellable graph, then $$is(G \vee \{u\}) = \alpha(G) \leq \chi(G^c) \leq t.$$
\end{lema}

We start this subsection by showing that 
\begin{equation}
\label{eq:7}
\chi(G^c) \leq T(G\vee\{u\}) \leq \chi(G^c)+\Delta(G)+1.    
\end{equation}
The lower bound is given by Lemma~\ref{lemachilocal}. 
The upper bound holds because we can use $\chi(G^c)$ tessellations to cover a partition $P$ of the vertices of $G$ in cliques $p_1, p_2, \ldots, p_i$ by assigning a different tessellation for each $p_j$ for $1 \leq j \leq i$. 
Moreover, the edges between $u$ and the vertices of $p_j$ maintain the same tessellation of $p_j$. 
The remaining edges between vertices of $p_1, p_2, \ldots, p_i$ are covered by cliques of size two with the non used $\Delta(G)$+1 tessellations following an edge coloring of such edges.
Thus, there is no universal graph such that the gap between $T(G\vee\{u\})$ and $\chi(G^c)$ is larger than $\Delta(G)+1$.
In particular, if $\chi(G^c) \geq 2\Delta(G)+1$, then by Theorem~\ref{lema:chiGc} below $T(G\vee\{u\}) = \chi(G^c)$.

\begin{teo}
\label{lema:chiGc}
A graph $G \vee \{u\}$ with $\theta(G) \geq 2\Delta(G)+1$ has $T(G\vee\{u\}) = \theta(G)$.
\end{teo}
\begin{proof}
Note that $\theta(G)=\chi(G^c)$.
Consider a graph $G\vee\{u\}$.
By Lemma~\ref{lemachilocal}, $T(G\vee\{u\}) \geq \chi(G^c)$.
Now we prove that $T(G\vee\{u\}) \leq \chi(G^c)$. 
Since $\chi(G^c) \geq 2\Delta(G)+1$, there is a tessellation cover of $G\vee\{u\}$ with $\chi(G^c)$ tessellations as follows.
Use the $\chi(G^c)$-coloring of $G^c$ as a guide to define a partial tessellation of the edges of $G$ with $\chi(G^c)$ tessellations in so that each color class of $G^c$ induces a clique of $G$ entirely covered by the tessellation related to this color class. 
Moreover, we extend the tessellations so that the edges $uw$ are covered by the corresponding tessellation of the color class of $w$. Now, the remaining edges in $G\vee\{u\}$ are the ones between vertices of $G$ that are not in a same color class in the coloring of $G^c$. The maximum number of tessellations incident to the endpoints of an edge $xy$ is $2\Delta(G)$ because $2$ tessellations come from the edges $ux$ and $uy$, and $2\Delta(G)-2$ come from the edges of $G$ incident to $x$ and~$y$.
Therefore, $T(G\vee \{u\}) \leq \chi(G^c)$ because it is possible to greedily assign  tessellations of cliques of size two to these edges.

\end{proof}

\begin{cor}
\label{cor:chigc}
A graph $G \vee \{u\}$ with $is(G \vee \{u\}) = \alpha(G\vee\{u\}) \geq 2\Delta(G)+1$ has $T(G\vee\{u\}) =  \chi(G^c)$. Moreover, if $H$ is a ($2\Delta(G)+1, G$)-graph of Construction~\ref{cons:1} on $G$ with $2\Delta(G)+1$ pendant vertices to $u$, then $T(H) = \theta(G) = \chi(G^c) + 2\Delta(G)+1$.
\end{cor}
\begin{proof}
Note that if $\alpha(G\vee\{u\}) \geq 2\Delta(G)+1$, then $ \chi(G^c) \geq \omega(G^c) = \alpha(G \vee \{u\}) \geq 2\Delta(G)+1$ and, by Theorem~\ref{lema:chiGc}, $T(G\vee\{u\}) =  \chi(G^c)$.
Consider now the graph $H$.
Since each pendant vertex added to $u$ in $H$ does not modify $\Delta(G)$, and it increases $is(H)$ by one unit, $is(H)\geq 2\Delta(G)+1$.
Moreover, each of those pendant vertices is a universal vertex in $G^c$, increasing  $\chi(G^c)$ by one unit.
Thus, $T(H) = \chi(H^c) = \chi(G^c) + 2\Delta(G)+1$.
\end{proof}

\subsection*{Good tessellable universal graphs}

A universal graph $G\vee \{u\}$ is good tessellable if $T(G\vee \{u\})=is(G\vee \{u\})$.
In this case, by Lemma~\ref{lemachilocal}, $T(G\vee\{u\})=\chi(G^c) = is(G\vee \{u\})$.
Therefore, if $G \vee \{u\}$ has $T(G \vee \{u\}) > \chi(G^c)$, then it is not a good tessellable graph.
By Corollary~\ref{cor:chigc}, if $\alpha(G \vee \{u\}) \geq 2\Delta(G)+1$, then
$T(G \vee \{u\}) = \chi(G^c)$, and $G \vee \{u\}$ is good tessellable when $\chi(G^c)=\omega(G^c)=is(G \vee \{u\})$.

The computational complexity of \textsc{gtr} of a subclass of universal graphs depends on the restrictions used to define the subclass.
On the one hand, perfect graphs $G$ with $\alpha(G) \geq 2\Delta(G)+1$ can be recognized in polynomial time~\cite{ArPerfeitoo}, and the addition of a universal vertex results in a good tessellable universal graph. On the other hand, planar graphs $G$ with  $\Delta(G) \leq 4$ and $\alpha(G) \geq 2\Delta(G)+1=9$ for which to decide whether $\chi(G)=\omega(G)=3$ is $\mathcal{NP}$-complete~\cite{ArGarey}.


\subsection*{Graphs with arbitrary gap between $T(G)$ and $is(G)$}

We start by showing that the gap between $T(G)$ and $is(G)$  can be arbitrarily large for graphs $G$ composed by the join of the complement of Mycielski graphs with a vertex $u$.

Since the Mycielski graph $M_j$ is triangle-free~\cite{BkWest}, the graph $M_j^{c}$ has no independent set of size three and $is(M_j^{c} \vee \{u\}) = 2$.  Moreover, $\chi(M_j)=j$~\cite{BkWest}, and by Lemma~\ref{lemachilocal},  $T(M_j^{c} \vee \{u\}) \geq \chi((M_j^c)^c) \geq j$. Fig.~\ref{fig:mycielski} depicts an example of the Mycielski graph $M_4$ and the relation between its $4$-coloring and a minimal tessellation cover of $M_4^{c} \vee \{u\}$. Therefore, there is a graph $H=M_j^{c} \vee \{u\}$ with $is(H) = 2$ and $T(H) \geq j$ for $j \geq 3$.

\begin{figure}
\centering
     \includegraphics[scale=0.60]{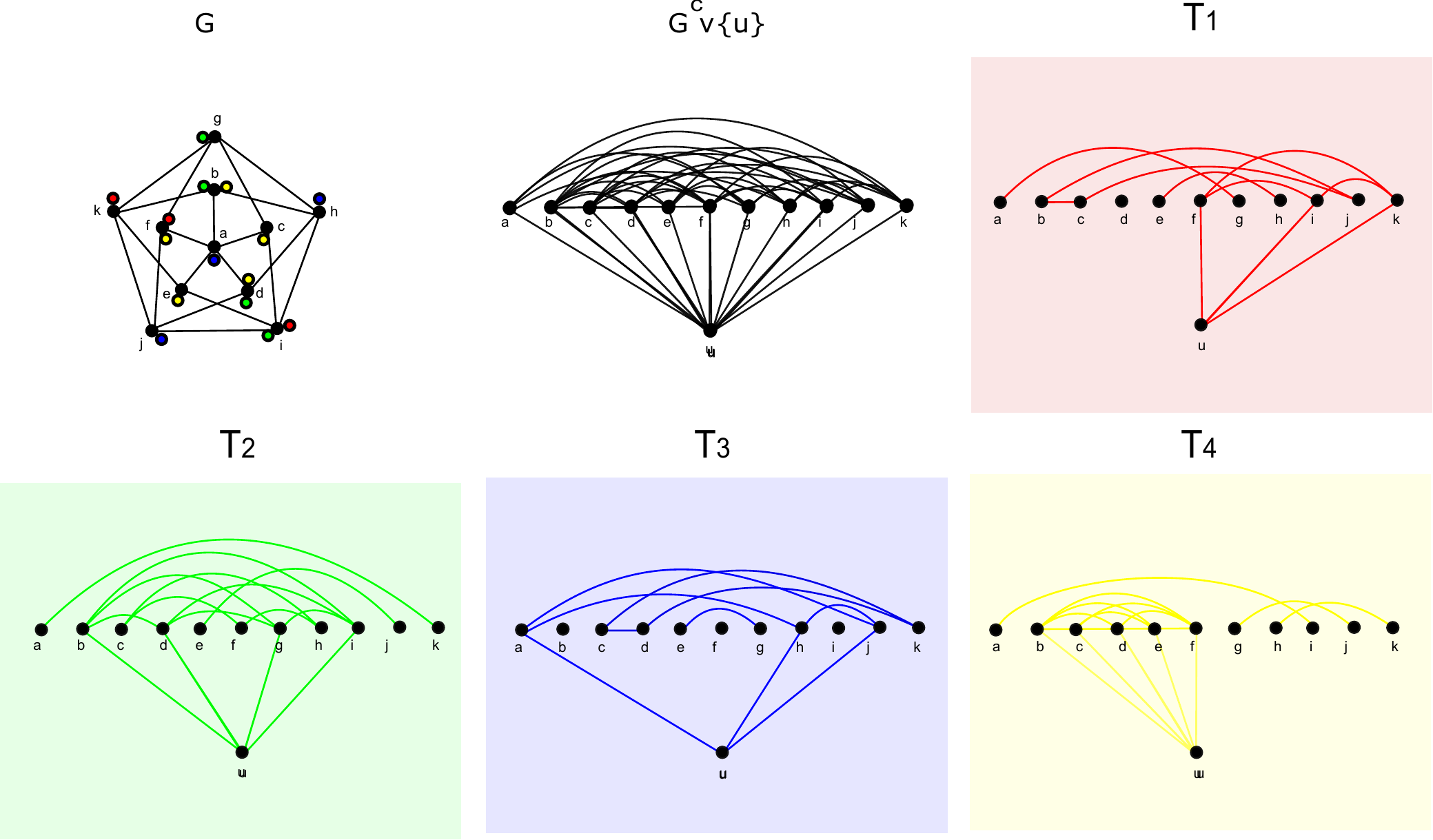}
     \caption{A Tessellation cover of $M_j^{c} \vee \{u\}$ with $4$ tessellations and possible $4$-colorings of $M_4$ guided by this tessellation cover. \label{fig:mycielski}}
\end{figure}

Now, we describe a subclass of universal graphs for which the gap between $T(G)$ and $is(G)$ is very large. 
We also show that \textsc{$k$-star size} and \textsc{$t$-tessellability} are $\mathcal{NP}$-complete for graphs of Construction~\ref{cons:classe}, for which \textsc{gtr} is in $\mathcal{P}$.

\begin{cons}
\label{cons:3}
\emph{
Let $G=(V,E)$ be a graph.
Obtain $S_2(G)$ by subdividing each edge of $G$ two times, so that each edge $vw \in E(G)$ becomes a path $v,x_1,x_2,w$, where $x_1$ and $x_2$ are new vertices. Let $L(S_2(G))$ be the line graph of $S_2(G)$. Add a universal vertex $u$ to $L(S_2(G))$, that is, consider the graph $L(S_2(G))\vee \{u\}$.
}
\end{cons}

First, we show that there is a connection between $T(H)$ of a graph $H$ of Construction~\ref{cons:3} on $G$ with the size of a maximum stable set of $G$.

\begin{teo}\label{thm:LvI}
If $G=(V, E)$ is a graph with $|E(G)|\geq 4$ and $H=(L(S_2(G)) \vee  \{u\})$ is obtained from Construction~\ref{cons:3} on $G$, then $T(H)=|V(G)| + |E(G)| - \alpha(G)$.
\label{teo:23a}
\end{teo}
\begin{proof}
We claim that $T(H) = \chi((H \setminus \{u\})^c)$. 
By Lemma~\ref{lemachilocal}, $T(H)\geq \chi((H \setminus \{u\})^{c})$.
Now, we show that $T(H)\leq \chi((H \setminus \{u\})^{c})$. 
Consider a partial tessellation cover of $H \setminus \{u\}$ with $\chi((H \setminus \{u\})^{c})$ tessellations induced by a coloring of $(H \setminus \{u\})^{c}$,  
so that cliques of $H \setminus \{u\}$ are assigned to tessellations associated to different colors 
of $(H \setminus \{u\})^{c}$. 
Since $H \setminus \{u\}=L(S_2(G))$ is the line graph of a $S_2(G)$ graph, every vertex of $(H \setminus \{u\})^{c}$ has a maximum clique of size two and another maximum clique incident to it with an arbitrary size.
Consider now a maximum clique $K_a$ of size at least three which 
is not completely covered yet.
The partial tessellation cover cannot have two cliques completely inside $K_a$ (otherwise their merge would result in a coloring of the complement graph with less than its chromatic number).
Therefore, the edges of $K_a$ are partially covered at this moment with one clique, and the remaining cliques covering the vertices of $K_a$ are the maximal cliques of size two that are incident to the vertices of $K_a$. 
Thus, if the partial tessellation cover of $K_a$ has only maximal cliques of size two given by edges incident to $K_a$, then each edge $e$ of $K_a$ has at most two already used tessellations on cliques incident to their endpoints (the ones given to these maximal cliques of size two).

Poljak~\cite{ArPoljak} proved that
$\chi(L(S_2(G))^{c}) = |V(G)| + |E(G)| - \alpha(G)$.
Since $|E(G)| \geq 4$ and $\alpha(G) \leq |V(G)|$, $|V(G)| + |E(G)| - \alpha(G) \geq 4$ and there is at least one available tessellation for each edge of $K_a$.
We claim that these available tessellations for each edge are enough to extend this partial tessellation cover to all edges of $K_a$.
First, pick an arbitrary tessellation for each edge.
Since the endpoint vertices of any collection of edges of $K_a$ on a same available tessellation do not have these tessellations incident to their endpoints, we cover the clique induced by these vertices with this tessellation.

Otherwise, $K_a$ has a clique $K_b$ in the partial tessellation cover and all the other vertices of $K_a$ must be covered by maximal cliques of size two with edges outside $K_a$. 
So, modify this partial tessellation cover assigning the tessellation of $K_b$ into all edges of $K_a$ and remove the vertices of $K_a$ from cliques of size two on this partial tessellation cover, i.e., now they are cliques of size one and $K_a$ is entirely covered by the tessellation of $K_b$.

The remaining uncovered edges of $H \setminus \{u\}$ in this partial tessellation cover are maximal cliques of size two.
Now, if an edge is uncovered and it is incident to a maximal clique of size two or more, then we need this clique to be entirely covered by a single tessellation.
Therefore, the maximum number of already used tessellations incident to the endpoints of a remaining edge is three.
Since there are $|V(G)| + |E(G)| - \alpha(G) \geq 4$ tessellations, there is always an available tessellation for these edges. 
Finally, the edges incident to $u$ are covered by the tessellations which covered the cliques of $H \setminus \{u\}$. 
Then, $T(H)\leq \chi((H \setminus \{u\})^{c})$.
\end{proof}

\begin{figure}
\centering
     \includegraphics[scale=0.4]{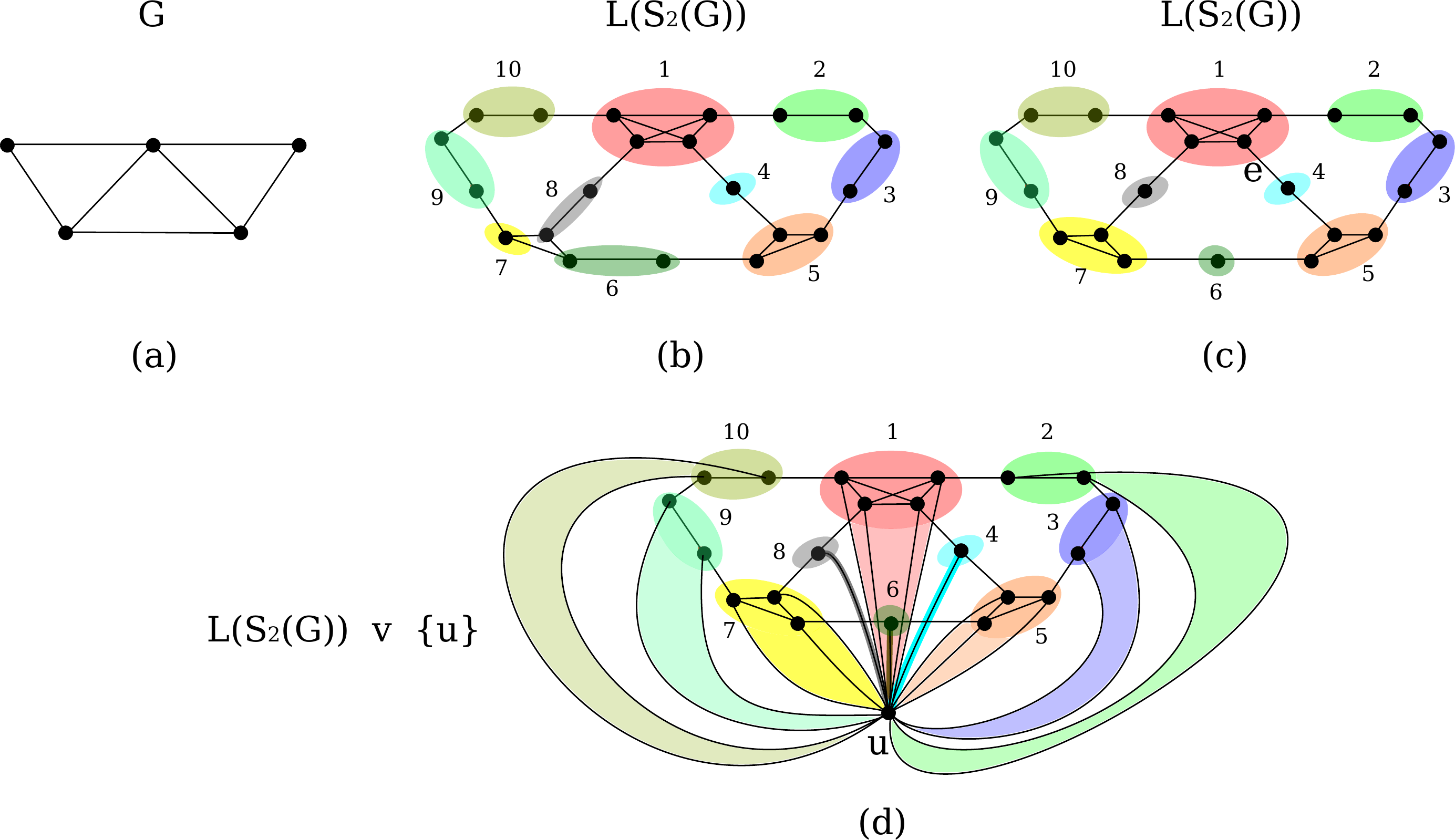}
     \caption{A tessellation cover of $H=L(S_2(G)) \vee \{u\}$ with $|V(G)|+|E(G)|-\alpha(G)$ tessellations. \label{fig:universal3}}
\end{figure}

Figure~\ref{fig:universal3} depicts the proof of Theorem~\ref{teo:23a}. In (a), we have graph $G$. 
In (b), we have a clique cover of $L(S_2(G))$.
In (c), we modify the clique cover so that the clique with label $7$ is covered by a new tessellation and at the same time we remove the vertices of the cliques of size two incident to the clique with label $7$. Now the cliques with labels $6$ and $8$ have only one vertex each. 
Finally, in (d) we extend the partial tessellation cover of $L(S_2(G))$ to include the edges incident to $u$.

Since deciding whether $\alpha(G) \geq k$ is $\mathcal{NP}$-complete~\cite{ArGarey}, by Theorem~\ref{teo:23a} we have the following result for the graphs of Construction~\ref{cons:3}.

\begin{cor}
\textsc{$t$-tessellability} is $\mathcal{NP}$-complete for universal graphs. 
\label{teo:23b}
\end{cor}
 \begin{proof}
Let $G$ be an instance graph of \textsc{$k$-independent set} with $|E(G)| \geq 4$. 
 We know that deciding whether $\alpha(G) \geq k$ is $\mathcal{NP}$-complete~\cite{ArGarey}.
 Consider the graph $H$ of Construction~\ref{cons:3} on $G$ with
 $H=L(S_2(G)) \vee \{u\}$.
 By Theorem~\ref{teo:23a},  $T(H) = |E(G)| + |V(G)| - \alpha(G)$.
 Therefore, deciding whether $\alpha(G) \geq k$ is equivalent to decide whether $T(H) \leq t = |E(G)| + |V(G)| - k$.
 \end{proof}

Next, we show that there are graphs of Construction~\ref{cons:3} for which the gap between $T(G)$ and $is(G)$ is very large, whereas \textsc{$t$-tessellability} remains $\mathcal{NP}$-complete.

\begin{teo}
\label{teo:afastarnpc}
Let $G$ be a graph and $H=L(S_2(G')) \vee \{u\}$ be obtained from Construction~\ref{cons:3} on $G'$, where $G'$ is obtained from $G$ by adding $x$  universal vertices, with $x$ polynomially bounded by the size of $G$.
To decide whether $T(H)=k$ with $k \geq is(H) +c$, for $c=(O|V(G)|^d)$ and constant $d$, is $\mathcal{NP}$-complete.
\end{teo}
\begin{proof}
Consider a graph $G$ and $L(S_2(G))\vee\{u\}$ as described in Corollary~\ref{teo:23b}.
Note that $is(L(S_2(G))\vee\{u\}) = \alpha(L(S_2(G))) = \mu(S_2(G))$.
We claim that $\mu(S_2(G)) = |E(G)|+\mu(G)$.
There are three edges in $S_2(G)$ between two adjacent vertices of $G$. 
In a maximum matching of $S_2(G)$, we need to select at least one of them,
otherwise, we could include the middle edge to a maximum matching, which is a contradiction.
Moreover, if there is only one edge and it is not a middle edge, then we obtain another maximum matching by replacing this edge by the middle edge.
Clearly, we cannot choose three edges and in case we choose two edges, different from the middle edge.
The case of two edges forces that both of them are incident to vertices of $S_2(G)$ associated to vertices of $V(G)$.
Therefore, the maximum number of such selection of two edges in $S_2(G)$  is equal to the size of a maximum matching of $G$.
For each edge in a maximum matching $\mu(G)$ of $G$ we have two edges in the maximum matching in $S_2(G)$ and, for each other edge of $G$, we have one edge in the maximum matching of $S_2(G)$.
Thus, $\mu(S_2(G)) = 2\mu(G) + |E(G)| - \mu(G) = |E(G)| + \mu(G)$.

By Theorem~\ref{teo:23a}, $T(H) = |E(G)| + |V(G)| - \alpha(G)$.
The addition of universal vertices to $G$ does not modify $\alpha(G)$.
The addition of each universal vertex may add one unit to $\mu(G)$ until it reaches $|V(G)|$.
Then, we add one unit to $\mu(G)$ for each addition of two universal vertices.
In that case, we start to increase the difference between $T(H) = |E(G)|+|V(G)| - \alpha(G)$ and $is(H) = |E(G)| + \mu(G)$, since for each two universal vertices we add to $G$, we increase $T(H)$ by two units and $is(H)$ by one unit.
Therefore, we can arbitrarily enlarge the gap between $T(G)$ and $is(G)$.
And, as long as the addition of these universal vertices are polynomially bounded by the size of $G$, it holds the same polynomial transformation of Corollary~\ref{teo:23b} from \textsc{$k$-independent set} of $G'$ (obtained from $G$ by the addition of universal vertices) to \textsc{$t$-tessellability} of $L(S_2(G'))\vee\{u\}$.
\end{proof}

Finally, we show that the graphs from Construction~\ref{cons:classe} below obey behavior~(e).

\begin{cons}
\label{cons:classe}
Let $H_1$ be the graph obtained from Construction~\ref{cons:2} (I) on a given graph $G_1$ and a non-negative integer $i$.
Let $H_2$ be the graph obtained from Construction~\ref{cons:3} on the graph $G_2 \vee K_{3|V(G_1)|}$ of a given graph $G_2$.
Let $u$ and $u'$ be the two universal vertices of the two connected components of $H_1$.
Add $is(H_2)$ degree-1 vertices to $H_1$ adjacent to $u$ and $is(H_2)$ degree-1 vertices adjacent to $u'$.
Consider $H_1 \cup H_2$.
\end{cons}

\begin{teo}\label{corFigD}
\textsc{$k$-star size} and \textsc{$t$-tessellability} are $\mathcal{NP}$-complete for graphs of Construction~\ref{cons:classe}, for which \textsc{gtr} is in $\mathcal{P}$.
\end{teo}
\begin{proof}
Let $G_1$ be an instance graph with no universal vertex of the well-known $\mathcal{NP}$-complete problem \textsc{$q$-colorability}~\cite{ArGarey}.
Let $G_2$ be an instance graph of the well-known $\mathcal{NP}$-complete problem \textsc{$p$-independent set} with $E(G_2) \geq 4$~\cite{ArGarey}.
Consider a graph $H=H_1 \cup H_2$ obtained from Construction~\ref{cons:classe} on $G_1$ and $G_2$ with $i=q$.

Since $H_2$ is obtained from Construction~\ref{cons:3} on $G_2 \vee K_{3|V(G_1)|}$, by Theorem~\ref{teo:afastarnpc}, $T(H_2) - is(H_2) > |V(G_1)|$.
By Theorem~\ref{thm:fraction}, $1 \leq is(H_1) \leq T(H_1) = |V(G_1)|$.
The parameter $is(H_2)$ can be obtained in polynomial time by applying a maximum matching algorithm~\cite{ArGarey} (see Theorem~\ref{teo:afastarnpc}).
And the addition of the degree-1 vertices to $H_1$ of Construction~\ref{cons:classe} implies that $1+is(H_2) \leq is(H_1) \leq T(H_1) = |V(G_1)|+is(H_2)$.

Therefore, $H=H_1 \cup H_2$ is a graph that obeys $is(H_2) \leq is(H_1) \leq T(H_1) \leq T(H_2)$ with $T(H)=T(H_2)$ and $is(H) = is(H_1)$.
The proof holds because \textsc{gtr} is in $\mathcal{P}$ with answer always no and both \textsc{$k$-star size} on graphs $H_1$ of Construction~\ref{cons:2} (I) (see Theorem~\ref{thm:fraction}) 
and \textsc{$t$-tessellability} on graphs $H_2$ of Construction~\ref{cons:3} (see Theorem~\ref{teo:afastarnpc}) are $\mathcal{NP}$-complete.
\end{proof}


\section{Concluding remarks}\label{sec:conc}

The concept of tessellation cover of graphs appeared in a thesis by Duchet~\cite{duchet1979representations}, and subsequently in~\cite{blokhuis1995equivalence, esperet2010covering}, as \textsc{equivalence covering}. The known results about tessellation cover number of a graph up to now were related to upper bounds of the values of $T(G)$, and the complexities of the \textsc{$t$-tessellability} problem~\cite{ArLatin}. In this work we focus on a different approach by analyzing the tessellation cover number $T(G)$ with respect to $is(G)$, one of its lower bounds, which implicitly appeared in the previous hardness proofs of~\cite{ArLatin}. 

The motivation to define the tessellation cover number comes from the analysis of the dynamics of quantum walks on a graph $G$ in the context of quantum computation~\cite{PSFG16}. Since it is advantageous to implement physically as few operators as possible in order to reduce the complexity of the quantum system, it is important to analyze the gap between $T(G)$ and $is(G)$. 

We have proposed the \textsc{good tessellable recognition} problem \textsc{(gtr)}, which aims to decide whether a graph $G$ satisfies $T(G) = is(G)$, and
we have analyzed the combined behavior of the computational complexities of the problems \textsc{$k$-star size}, \textsc{$t$-tessellability}, and \textsc{gtr}. We have defined graph classes corresponding to triples which specify the computational complexities of these problems, summarized in Table~\ref{tab:classes}. 
We have defined graph classes in Construction~\ref{cons:2}~(I) and Construction \ref{cons:2}~(II) that obey behaviors ($\mathcal{NP}$-complete, $\mathcal{P}$, $\mathcal{NP}$-complete) and ($\mathcal{NP}$-complete, $\mathcal{P}$, $\mathcal{P}$), respectively. Graphs that obey behavior~($\mathcal{NP}$-complete, $\mathcal{NP}$-complete, $\mathcal{P}$) are obtained using Construction~\ref{cons:classe}. 
We also note that there are omitted triples in Table~\ref{tab:classes}, which are either empty or easy to provide examples, as described in Section~\ref{sec:intro}. 

We are interested in the following two research topics: 
(i) The concept of good tessellable graphs can be extended to \emph{perfect tessellable graphs}, the graphs $G$ for which $T(H) = is(H)$ for any induced subgraph $H$ of $G$. 
A natural open task is to establish the characterization by forbidden induced subgraphs and a polynomial-time recognition algorithm for perfect tessellable graphs.
We conjecture that this class is exactly the \{gem, $W_4$, odd cycles\}-free graphs;
(ii) We have already established relations between $T(G)$ with other well-known graph parameters such as the chromatic number and the maximum size of a stable set. 
We are currently investigating further relations such as those between $T(G)$ with the chromatic index and the total chromatic number.

\bibliographystyle{plain}
\bibliography{refs}

\begin{thebibliography}{10}

\bibitem{ArLAWCG}
A.~Abreu, F.~Marquezino, and D.~Posner.
\newblock Tessellations on graphs with few \uppercase{P}$_4$'s.
\newblock {\em Proceedings of VIII Latin American Workshop on Cliques in
  Graphs, LAWCG}, page~8, 2018.

\bibitem{ArLatin}
Alexandre Abreu, Luis Cunha, Tharso Fernandes, Celina de~Figueiredo, Luis
  Kowada, Franklin Marquezino, Daniel Posner, and Renato Portugal.
\newblock The graph tessellation cover number: extremal bounds, efficient
  algorithms and hardness.
\newblock {\em Proceedings of the 13th Latin American Theoretical Informatics
  Symposium, LNCS}, 10807:1--13, 2018.

\bibitem{blokhuis1995equivalence}
Aart Blokhuis and Ton Kloks.
\newblock On the equivalence covering number of splitgraphs.
\newblock {\em Inf. Process. Lett.}, 54(5):301--304, 1995.

\bibitem{duchet1979representations}
Pierre Duchet.
\newblock {\em Repr{\'e}sentations, noyaux en th{\'e}orie des graphes et
  hypergraphes}.
\newblock PhD thesis, Univ. Paris 6, 1979.

\bibitem{esperet2010covering}
Louis Esperet, John Gimbel, and Andrew King.
\newblock Covering line graphs with equivalence relations.
\newblock {\em Discrete Appl. Math.}, 158(17):1902--1907, 2010.

\bibitem{ArGarey}
Michael~R. Garey and David~S. Johnson.
\newblock {\em Computers and Intractability: A Guide to the Theory of
  NP-Completeness}.
\newblock W. H. Freeman Co, 1979.

\bibitem{ArPerfeitoo}
M.~Grotschel, L.~Lovasz, and A.~Schruver.
\newblock The ellipsoid method and its consequences in combinatorial
  optimization.
\newblock {\em Combinatorica}, 1(2):169--197, 1981.

\bibitem{ArClasse1}
M.~Plantholt.
\newblock The chromatic index of graphs with a spanning star.
\newblock {\em J. Graph Theory}, 5:45--53, 1981.

\bibitem{ArPoljak}
S.~Poljak.
\newblock A note on stable sets and colorings of graphs.
\newblock {\em Comment. Math.}, 15:307--309, 1974.

\bibitem{PSFG16}
R.~Portugal, R.~A.~M. Santos, T.~D. Fernandes, and D.~N. Gon{\c{c}}alves.
\newblock The staggered quantum walk model.
\newblock {\em Quantum Inf. Process}, 15(1):85--101, 2016.

\bibitem{ArCNMAC}
D.~Posner, C.~Silva, and R.~Portugal.
\newblock On the characterization of 3-tessellable graphs.
\newblock {\em Proceedings Series of the Brazilian Society of Computational and
  Applied Mathematics, CNMAC}, 6(2):1--7, 2018.

\bibitem{BkWest}
D.~West.
\newblock {\em Introduction to Graph Theory}.
\newblock Pearson, 2000.

\end{thebibliography}

\end{document}